\documentclass[conference]{IEEEtran}

\usepackage[dvips]{color}
\usepackage{epsf}
\usepackage{times}
\usepackage{epsfig}
\usepackage{graphicx}

\usepackage{amsmath}
\usepackage{amssymb}
\usepackage{amsxtra}
\usepackage{amsthm}
\usepackage{bbm}

\usepackage{here}
\usepackage{rawfonts}
\usepackage{times}
\usepackage{url}
\usepackage{cite}
\usepackage{comment}
\usepackage[utf8]{inputenc}
\usepackage{caption}
\usepackage{subcaption}

\usepackage{pstricks}
\usepackage{algorithm}
\usepackage{algpseudocode}
\usepackage{lipsum}

\DeclareMathOperator\Arg{Arg}

\newtheorem{proposition}{\bf Proposition}

\makeatletter
 
\IEEEoverridecommandlockouts

\begin{document}

\title{Reflections in the Sky: Millimeter Wave Communication with UAV-Carried Intelligent Reflectors }
	\vspace{-0.2cm}	
\author{\IEEEauthorblockN{Qianqian Zhang$^1$, Walid Saad$^1$, and Mehdi Bennis$^2$ } 	
	\IEEEauthorblockA{\small
		$^1$Bradley Department of Electrical and Computer Engineering, Virginia Tech, VA, USA,
		Emails: \url{{qqz93,walids}@vt.edu}. \\
		$^2$Center for Wireless Communications, University of Oulu, Finland, Email: \url{mehdi.bennis@oulu.fi}.   
		\vspace{-0.5cm}		
}
}
\maketitle

\begin{abstract}
	In this paper, a novel approach that uses an unmanned aerial vehicle (UAV)-carried intelligent reflector (IR) is proposed to enhance the performance of millimeter wave (mmW) networks.  
	In particular, the UAV-IR is used to intelligently reflect mmW beamforming signals from a base station  towards a mobile outdoor user, while harvesting energy  from mmW signals to power the IR.  
	To maintain a line-of-sight (LOS) channel, a reinforcement learning (RL) approach, based on Q-learning and neural networks, is proposed to model the propagation environment, such that the location and reflection coefficient of the UAV-IR can be optimized to maximize the downlink transmission capacity.  
	Simulation results show a significant advantage for using a UAV-IR over a static IR, in terms of the average data rate and the achievable downlink LOS probability.  
	The results also show that  the RL-based deployment of the UAV-IR further improves the network performance, relative to a scheme without learning. 
\end{abstract}

\IEEEpeerreviewmaketitle

\section{Introduction}
Next-generation cellular systems will inevitably rely on high-frequency millimeter wave (mmW) communications in order to meet the growing need for wireless capacity \cite{saad2019vision}. 
However, communicating at mmW frequencies faces many challenges. 
One prominent such challenge is the high susceptibility of mmW links to  blockage caused by common objects, such as trees and human bodies, which can seriously attenuate mmW signals \cite{bennis2018ultrareliable}.  
Enabling reliable mmW links under blockage is therefore a major barrier hindering the  deployment of mmW bands into wide-scale commercial uses.

To overcome these drawbacks of mmW, signal reflectors have been recently proposed to bypass obstacles and prolong the communication range \cite{peng2015effective}. 
In particular, by using reflectors, a non-line-of-sight (NLOS) mmW link can be compensated by creating multiple, connected line-of-sight (LOS) links, thus significantly reducing the mmW channel attenuation.  
It has also been shown that the use of reflectors is more appropriate for mmW networks than conventional relay stations (RSs) \cite{tan2018enabling}.   
Different from traditional RSs that receive, amplify (or decode), and forward the mmW signal, a reflector only reflects the incident signal towards the receiver, by inducing a certain phase shift.  
Therefore, reflective relaying incurs no additional receiving noise, and is energy efficient.  
Due to the nature of reflective surfaces, the connected LOS links that pass through one or more reflectors can share the same frequency band, thus improving the spectrum efficiency.  
In particular, an intelligent reflector (IR) that consists of a large number of low-cost passive components can realize beamforming with little energy cost \cite{boyer2014invited}. 
Each IR component reflects the incident signals, while harvesting the radio frequency (RF) energy from the unreflected fraction of signals to power itself. 
By collaboratively adjusting the phase shifts, an IR can focus the reflected signal into a sharp beam, 
hence maximizing beamforming performance gains.  

Several recent works have studied the use of IRs to enhance the performance of cellular networks such as in  \cite{peng2015effective} and \cite{hong2018mmwave,wu2018intelligent,barros2017integrated,huang2018energy}.  
In \cite{peng2015effective}, the design of a passive reflector and the estimation of the reflection gain are presented for urban mmW communications. 
In \cite{hong2018mmwave}, simulations and experiments are carried out in indoor environments to measure the transmission of passive reflectors over mmW.  
The authors in \cite{wu2018intelligent} jointly optimized the transmit beamforming from 
an access point and the reflective beamforming in IR to maximize the received signal power at the user equipment (UE).  
The authors in \cite{barros2017integrated} conducted propagation modeling of terahertz signals using reflectors. 
Meanwhile, the work in \cite{huang2018energy} studied the energy efficiency of reflector-assisted downlink communications. 
However, prior works on mmW reflectors in \cite{peng2015effective} and \cite{hong2018mmwave} focus mainly on experimental measurements, while the IR-related works in \cite {wu2018intelligent,barros2017integrated,huang2018energy} study cellular communications over non-mmW spectrum.   
Furthermore, all of the prior works in \cite{peng2015effective} and \cite{hong2018mmwave, wu2018intelligent, barros2017integrated, huang2018energy} rely on passive reflectors placed at a fixed location, which cannot cope with the dynamic changes of mmW channels.  
For example, a simple body movement of the UE can cause significant mmW blockage  and render a static IR ineffective.    

Due to the blockage-prone nature of mmW signals, mobile reflectors are more appropriate to enhance mmW communications than stationary IRs.  
For instance, one can employ a UAV-carried IR (UAV-IR) that can adjust the location of the reflector constantly, according to the changes in environment thus maintaining persistent LOS links with both transmitter and receiver.   
UAV-assisted communications have attracted significant recent attention \cite{mozaffari2018tutorial,chen2017caching,mozaffari2017mobile}, however, no prior work has studied the use of UAV-IRs.  
In particular, a UAV-IR differs from a UAV-aided RS, due to its simpler antenna structure and smaller power supply, which make it more suitable for a dense communication scenario. 
Being embedded onto a UAV, an IR can improve the reliability of mmW transmissions  by  optimizing its location intelligently. 
However, such optimization  requires a precise channel state information (CSI) of the IR-UE link.   
Considering the possible motion of UAVs and UEs, as well as the blockage effect of the human body on mmW signals, the real-time value of CSI is difficult to obtain.   
Thus, to enable an efficient deployment of a UAV-IR for mmW transmissions,  
the challenge of CSI estimation must be properly addressed.

The main contribution of this paper is, thus, a novel framework for effective deployment of a UAV-IR to assist mmW downlink transmission in a dynamic environment with  moving UEs.  
To maintain a LOS channel, a reinforcement learning (RL) approach \cite{sutton2018reinforcement}, based on Q-learning and neural networks, is proposed to model the propagation environment, such that the location and reflection coefficient of the UAV-IR can be optimized to maximize the downlink transmission capacity.    
Meanwhile, we propose the use of RF energy harvesting for self powering the IR. 
Simulation results show the effectiveness of the proposed UAV-IR-based approach compared to a static IR. 
To the best of our knowledge, this is the first paper that proposes a learning-based deployment  of UAV-IRs for mmW communications with RF energy harvesting. 

The rest of this paper is organized as follows. Section \ref{sysModel_proFormulation} presents the system model and problem formulation. The framework of the optimal deployment of UAV-IR is proposed in Section \ref{solution}. Simulation results are presented in Section \ref{simulation}, while conclusions are drawn in Section \ref{conclusion}. 

\section{System Model and Problem Formulation}\label{sysModel_proFormulation}

\begin{figure}[!t]
	\begin{center}
		\vspace{-0.6cm}
		\includegraphics[width=8.5cm]{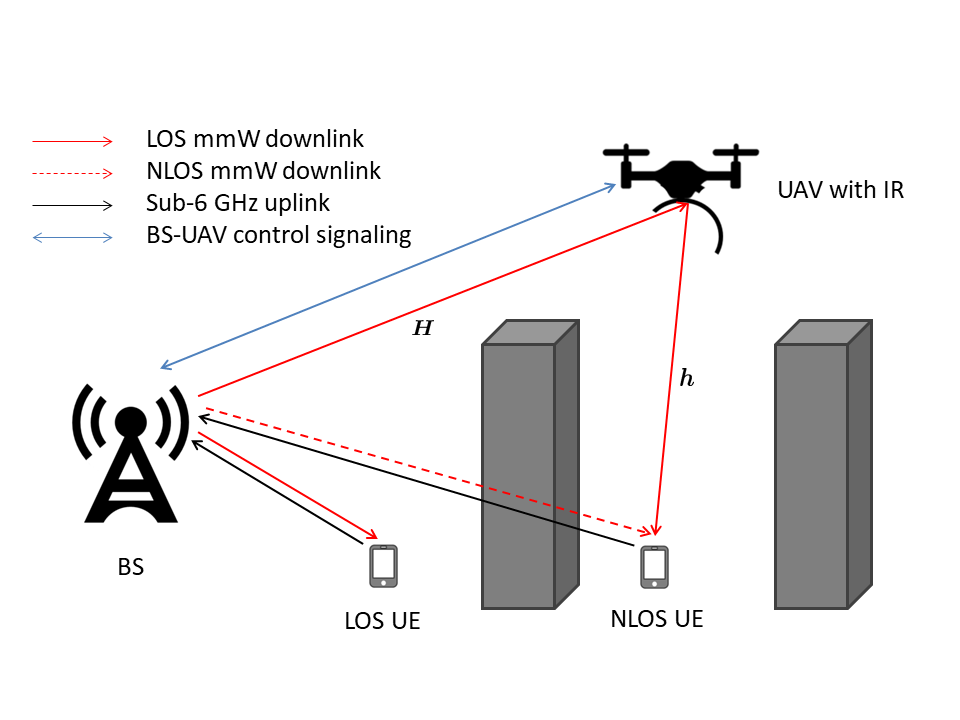}
		\vspace{-0.5cm}
		\caption{\label{systemmodel}\small If the downlink channel over mmW between the BS and a UE is blocked, a UAV-IR will be deployed to serve the NLOS UE. }
	\end{center}\vspace{-0.5cm}
\end{figure}

Consider a cellular base station (BS) serving a group of outdoor UEs via mmW frequencies over the downlink.  
To compensate for the fast attenuation of mmW signals, the BS is equipped with directional antenna arrays for beamforming. 
In a LOS case, the downlink communication can be reliable and efficient.  
However, when blockage happens, the propagation attenuation of the mmW channel will substantially increase.  
To improve the received power, a UAV-IR will be deployed to assist in serving NLOS UEs.  

As shown in Fig. \ref{systemmodel}, an IR replaces the direct NLOS link with two connected LOS links, by reflecting the mmW signals from the BS towards the UE.    
By equipping the IR on a UAV, the UAV can adjust the IR's position to maintain LOS links with both the BS and the UE.   
The BS assigns the UAV with a specific channel for control signaling.  
Therefore, the UE's feedback can be obtained via the UE-BS and BS-UAV links to improve the parameter setting of IR reflections.  
The UAV has a limited on-board energy, which restricts its service time. 
In order to avoid drawing energy from the UAV, the IR can harvest energy from the unreflected fraction of mmW beamforming signals, and convert it into the electrical energy via rectifiers to power itself.     

The optimal beamforming at the BS side for IR-assisted communications has been investigated in \cite{wu2018intelligent}  and \cite{huang2018energy} without a UAV.  
Here, we focus on the optimal deployment of the UAV-IR to maximize the downlink capacity, 
while harvesting RF energy to power the reflector. 
To address this problem, we will first provide the model of IR-assisted downlink transmissions, propose the net energy consumption model for the UAV-IR, and then, present the problem formulation.

\subsection{Channel Model and Communications Capacity} \label{channelModel}

Consider a multiple-input-single-output wireless link from the BS, reflected by the UAV-IR, to a single-antenna UE. 
We assume that the BS is equipped with $M \in \mathbb{N}^{+}$ antennas, and the UAV-IR consists of $N \in \mathbb{N}^{+}$ reflective components. 
Therefore, the channels of the BS-IR link and the IR-UE link can be denoted as $\boldsymbol{H} \in \mathbb{C}^{N\times M}$ and $\boldsymbol{h} \in \mathbb{C}^{1\times N}$, respectively.  
We assume that the CSI $\boldsymbol{H}$ of the BS-IR link is fully known to both the BS and the UAV-IR. 
However, due to the susceptibility of mmW to blockage, the CSI $\boldsymbol{h}$ of the IR-UE link will be considered as a variable.   

At the UAV-IR, the incident beamforming signal from the BS can be defined by an $N \times 1$ vector as 
$	\boldsymbol{r} = \boldsymbol{H} \boldsymbol{w}s$, 
where  $s$ represents the encoded symbol with zero mean and unit variance, and $\boldsymbol{w} \in \mathbb{C}^{M\times 1}$ is the transmission vector of linear beamforming  at the BS. 
After receiving $\boldsymbol{r}$, each component in IR combines the multi-path signals, and reflects the combined signal by inducing a phase shift.  
Structurally, an IR component is composed by an adjustable antenna and a rectifier. By varying the impedance of the antennas, the mismatch between the antenna structure with the carrier wave can reflect back a portion of the incident signal, while the remaining part can be harvested by the rectifier to assist powering the reflector. 
Let  $\mathcal{N}$ be the index set of reflective components in IR. 
For each  $n \in \mathcal{N}$, we denote the phase shift of reflection as $\theta_n \in [0,2\pi] $, and the amplitude reflection coefficient by $a \in [0,1)$.  
Then, the reflection coefficient of IR can be denoted by ${\Theta} = \text{diag} (a  e^{j\theta_1},\cdots, a e^{j\theta_N} )$, where  $\text{diag}(\boldsymbol{\cdot})$  is a diagonal matrix with each diagonal element  as $a  e^{j\theta_n}$.  

Let $\boldsymbol{x}\in \mathbb{R}^3$ be the location of the UAV-IR,  $\boldsymbol{y} \in \mathbb{R}^3$  be the UE's position, and $\omega \in [0, 2\pi)$ be the angle of the UE's position relative to the body of the human user.    
Due to the possible movement and motion of the UE, the values of $\boldsymbol{y}$ and $\omega$ may keep changing. 
For tractability, we assume that the real-time values of $\boldsymbol{y}$ and $\omega$  can be updated to the UAV, via a UE-BS uplink (control) channel and the BS-UAV control link. 
Then, the received signal at the UE will be  $\boldsymbol{h}(\boldsymbol{x},\boldsymbol{y}, \omega) {\Theta}  \boldsymbol{r} +z$,  
where $z$ is the receiver noise with an average noise power spectral density $n_0$. 
The signal-to-noise ratio (SNR) of the downlink transmission from the BS, reflected by the IR, to the served UE can be given by:  
\begin{align}\label{snr}
	\eta (\boldsymbol{x},\boldsymbol{y}, \omega, {\Theta}) =  \frac{|\boldsymbol{h}(\boldsymbol{x} ,\boldsymbol{y}, \omega) {\Theta}    \boldsymbol{r}  |^2 }{ bn_0},   
\end{align}   
where $b$ is the downlink bandwidth.  
Here, we define a minimum threshold  $\tau$ for the received SNR. If $\eta < \tau$, the UE cannot successfully decode any received signal.  
Consequently,  the capacity of the IR-assisted downlink communication can be expressed as 
\begin{align}\label{equCapacity}
	c (\boldsymbol{x},\boldsymbol{y}, \omega, {\Theta}) =  b \log_2 \left( 1+  \eta(\boldsymbol{x},\boldsymbol{y}, \omega, {\Theta}) \right).  
\end{align}

\subsection{Energy Consumption and Harvesting of UAV-IR} \label{subsecEnergyEff} 
To maintain a LOS link with the UE, the UAV must frequently adjust its location, according to the UE movement. 
Hence, the power consumption of a UAV-IR  mainly comes from the UAV's hovering and mobility, while  the adjustment of the IR's impedance loads for mmW reflections consumes  a small amount of energy.    
For tractability, we assume that the powers for hovering, mobility, and reflection are all constant values,  denoted by $p_h$, $p_m$ and $p_r$, respectively, where $p_m>p_h > p_r>0$.   
Moreover, we assume that the transmission only happens during the hovering state, 
since a moving IR will complicate the beamforming process at both the BS and the IR side.   
Thus, the BS will not transmit any downlink signal, until the location of the UAV-IR is fixed.   
Let  $v_r$ be the UAV's speed, and $\mathbbm{1}_{v_r=0}$ be an indicator function that equals to $1$ only when the UAV is hovering and its speed  $v_r$ equals zero.  

In the hovering state, while reflecting mmW signals, an IR can harvest the RF energy from the unreflected portion of the incident beamforming signal $\boldsymbol{r}$,  and  transfer it into electrical power via the rectifier.  
Let $\kappa \in (0,1)$ be the average energy converting efficiency.  
Then, the power harvested by IR from the beamforming signal $\boldsymbol{r}$ can be expressed by 
\begin{align}
	p_e( {\Theta}) = \kappa \| \left( \boldsymbol{I} - {\Theta}  \right) \boldsymbol{r} \|^2, 
\end{align}
where $\boldsymbol{I} \in \mathbb{R}^{N\times N}$ is the identity matrix.  
Considering the harvested power $p_e(\Theta)$ is much smaller compared with $p_h$ and $p_m$, it is mainly used to satisfy the reflector power $p_r$, such that the IR can self-power without drawing energy from the UAV.  
Therefore, the net power consumption of the UAV for providing a reflection service to a UE is 
\begin{align}
	p({\Theta},v_r) = \mathbbm{1}_{v_r=0} \cdot ( p_h + p_r- p_e({\Theta})) + (1-\mathbbm{1}_{v_r=0}) \cdot  p_m.  	 
\end{align}

\subsection{Problem Formulation} \label{problemFormulation}

We consider a dynamic process for the UAV-IR's deployment, in which the UAV-IR tracks the movement of the served UE, and provides downlink mmW transmissions via the IR. 
The coherence time of a wireless channel is given by: $\Delta t = \lambda_f/ v_{\text{e}} \frac{}{}$, where $\lambda_f$ is the carrier wavelength,  
and $v_{\text{e}}$ is the UE's speed.  
The CSI $\boldsymbol{h}$ within one coherence time $\Delta t$ is considered to be constant. 
Although the higher operating frequency at the mmW bands leads to a smaller wavelength $\lambda_f$, considering a lower velocity $v_{\text{e}}$ of the outdoor UE in our model, the coherence time $\Delta t$ of the mmW channel is similar to the current macrocell service using multiple antennas \cite{swindlehurst2014millimeter}.     
Furthermore, the measurement results in \cite{va2015basic} have shown a slower change of the mmW channel as the beamwidth of directional transmissions narrows.  
Therefore, it is possible to group multiple coherence times to form a longer interval $\Delta T = k\Delta t$, where $k \in \mathbb{N}^{+}$, and the UAV-IR deployment can be optimized for each time slot $\Delta T$, instead of $\Delta t$.

We divide the UAV-IR deployment into a sequence of two alternating stages: communication stage and mobility stage.  
In the communication stage, the UAV-IR is hovering and providing the downlink service towards the moving UE. 
However, once the UE receives an SNR $\eta$ lower than the threshold $\tau$,  blockage occurs in the downlink, and then, the UAV-IR will move to a new place and rebuild a LOS link for downlink transmissions.  
Note that, the length of each communication stage depends on the real-time state of the IR-UE link, while the duration of a mobility stage is determined by the travel distance and the speed of the UAV.  Thus, both stages can last for a multiple of $\Delta T$. 

Our goal is to jointly  optimize the location and reflection parameters of the UAV-IR, such that, before the onboard energy $E$ is exhausted, the total downlink transmissions that the UAV-IR provides to the  UE can be maximized, i.e.:   
\begin{subequations}\label{equsOpt}
	\begin{align}
		\max_{\boldsymbol{x},{\Theta}} \quad &  \sum_{t=1}^{T}  \mathbbm{1}_{v_r(t)=0} \cdot c\left(\boldsymbol{x}(t),\boldsymbol{y}(t), \omega(t),{\Theta}(t)\right) \cdot \Delta T  \label{equOpt}\\
		\textrm{s.t.} \quad  
		& v_r(t) = \frac{\|\boldsymbol{x}(t-1)-\boldsymbol{x}(t)\|}{\Delta T},\\
		& \sum_{t=1}^{T} p({\Theta}(t),v_r(t)) \cdot \Delta T \le E - \epsilon	, \\
		& 0 \le \theta_n(t) \le 2\pi,  \quad \forall n \in \mathcal{N}, 
	\end{align}
\end{subequations}
where  $T \in \mathbb{N}^{+}$ is a random integer which represents the end of the UAV's service,  and $\epsilon \in [0,p_m\Delta T)$ is the residual energy in the UAV at the end of its service.  
For notational convenience, hereinafter, we will only include the time index for each variable, but omit all the other parameters.  
 
For each slot $\Delta T$, the optimization of the reflection coefficient ${\Theta}(t)$ requires a precise CSI $\boldsymbol{h} (t)$ of the IR-UE link. 
Due to the possible movement of UE and UAV-IR and the blockage effect caused by the human body on mmW links, the CSI of the IR-UE link may keep changing, and thus, the real-time measurement is needed to determine the value of $\boldsymbol{h}(t)$.  
However, it is impractical to move the UAV to all possible locations to measure the CSI and jointly optimize the location and the reflection coefficient, because such sweeping search will consume  significant power and time.  
To address the challenge of the CSI estimation, a learning-based approach can be applied to enable an efficient deployment of the UAV-IR.

\section{Optimal Deployment of UAV-IR} \label{solution}

In this section, an RL framework, based on Q-learning and  neural networks, is proposed for the UAV-IR to find its optimal location, such that the CSI for the downlink mmW channels can be efficiently measured, and the reflective coefficient will be optimized accordingly to maximize the downlink transmission capacity.  
The model needs no prior knowledge of the dynamic environment, instead, it learns the property of the environment during the service process of the UAV-IR, based on the measurement and feedback during each communication stage.

\subsection{Stationary UE and UAV-IR}
Prior to learning, we first consider the simple case of a stationary UE with constant  $(\boldsymbol{y},\omega)$.    
In this case, once the UAV-IR fixes its location and measures the CSI of the IR-UE link, $\boldsymbol{h}$ is considered to be constant.    
Let $\Arg(\cdot)$ denote the argument of a complex number.
Then, we find the optimal reflection coefficient  ${\Theta}^{*} = \arg \max_{{\Theta}} c({\Theta})$ that maximizes the downlink capacity, given that the CSI $\boldsymbol{h}$ between the UAV-IR and the UE is known, as follows.    
\begin{proposition}\label{prop1} 
	For a known CSI $\boldsymbol{h}$, the optimal reflection parameter that maximizes the downlink transmission in (\ref{equOpt}) during each time slot is ${\Theta}^{*} =  \text{diag} (a  e^{j\theta_1^{*}},\cdots, a e^{j\theta_N^{*}} )$, where ${\theta}_n^{*} = -\Arg(h_n r_n)$, $\forall n \in \mathcal{N}$. 
\end{proposition} 
\begin{proof} 
	According to (\ref{snr}) and (\ref{equCapacity}), it is easy to check that $\max_{{\Theta}} c({\Theta}) \Longleftrightarrow \max_{{\Theta}} 	| \boldsymbol{h}{\Theta}\boldsymbol{r} |^2 $.  Meanwhile,  
	\begin{equation}\label{equ7}
		|\boldsymbol{h}{\Theta}\boldsymbol{r} | = |\sum_{n=1}^{N} a h_n r_n e^{j\theta_n} | \le \sum_{n=1}^{N} a |h_n r_n e^{j\theta_n}|. 
	\end{equation}
	The condition for equality in (\ref{equ7}) is that each complex component $h_n r_n e^{j\theta_n}$ has the same argument. 
	Without loss of generality, let $h_n r_n e^{j\theta_n} = 0, \forall n \in \mathcal{N}$.  
	Thus, we have ${\Theta}^{*} =  \text{diag} (a  e^{j\theta_1^{*}},\cdots, a e^{j\theta_N^{*}} )$, where $\theta_n^{*} = - \Arg(h_n r_n)$, $\forall n \in \mathcal{N}$. 
\end{proof} 
Therefore, given the CSI $\boldsymbol{h}$ between the UAV-IR and the served UE, the optimal reflection coefficient $\Theta^{*}$ can be uniquely determined. 
Hereinafter, $\Theta$ will be considered as its optimal value once $\boldsymbol{h}$ is available.

\subsection{Mobile UE and UAV-IR }

In the scenario of a moving UE, the CSI $\boldsymbol{h}$ of the IR-UE link will vary constantly. 
When downlink blockage happens,  
the UAV-IR must move to a new location and rebuild a LOS channel to serve the UE.   
In order to determine the optimal position for the UAV-IR of its next communication stage,  an RL framework, based on Q-learning and  neural networks, is proposed to learn and model the dynamic communication environment. 

In each communication stage, the UAV-IR will acquire the UE's location, and measure the CSI of the IR-UE link. 
After transmissions with an optimal reflection coefficient, the UAV-IR will get a feedback of the current downlink capacity, which will be used for the environment modeling. 
Once blockage occurs in the downlink, the optimal location, where the UAV-IR will be deployed during the next communication stage, can be determined, based on the current environment model.  
In the framework of  RL, the CSI $\boldsymbol{h}(t)$ of the IR-UE channel is called the \emph{environment}, the UAV is the \emph{agent} who can take \emph{action} $\boldsymbol{x}(t)$ that changes the environment state from one state to another, 
and a \emph{reward} is defined as the downlink data amount during the current time slot, as
\begin{equation}
	r(t) = \mathbbm{1}_{v_r(t)=0} c(t) \Delta T.
\end{equation}

However, since the UE is mobile, the environment can vary gradually, even without any action from the UAV. This environment is called \emph{nonstationary}.  
In order to understand the nonstationary property of $\boldsymbol{h}$ and identify the change caused by the human user,  
the UE's movement pattern must be properly modeled. 
Given a movement history $\boldsymbol{Y}_{t}=\{  (\boldsymbol{y}(0),\omega(0)),\cdots,(\boldsymbol{y}(t),\omega(t))\} $ during the past time slots, the probability that the UE will arrive at $(\boldsymbol{y},\omega)$ in the next $\Delta T$ can be formulated by  
$f (\boldsymbol{y},\omega|\boldsymbol{Y}_t) = N(\boldsymbol{\mu},\boldsymbol{\Lambda}|\boldsymbol{Y}_t)$,  
where $N(\cdot)$ denotes a Gaussian model with the mean $\boldsymbol{\mu} = \boldsymbol{y}(t)+ \boldsymbol{v}_{e}(t)\Delta T $ and the variance $\boldsymbol{\Lambda} = \text{Var}(\boldsymbol{Y}_t)$. 
Based on $f(\cdot|\boldsymbol{Y}_t)$, a Markov decision process (MDP) can be formulated to predict the  future movement of the served UE.  

In order to capture the action of the UAV-IR, a \emph{policy} $\pi_f( \boldsymbol{x} | \boldsymbol{h}, \boldsymbol{y},\omega)$ is introduced to represent the probability that the UAV moves towards $\boldsymbol{x}$,  given the IR-UE CSI $\boldsymbol{h}$, the location $(\boldsymbol{y},\omega)$ of the UE, and the prediction model $f$.    
The main objective of this learning approach is to determine the best policy, such that UAV can be deployed onto the optimal location $\boldsymbol{x}$ that maximizes the downlink transmission over a long term.  
In order to quantify the potential of each location $\boldsymbol{x}$ to provide a reliable downlink service, we define the value function for time $t$  as  
\begin{equation}\label{Qfunction}
	\begin{aligned}
		Q_t(\boldsymbol{h}(t),&\boldsymbol{x}(t+K), \boldsymbol{y}(t), \omega(t)) =    \\  
		&\mathbb{E} \left( \sum_{i=t+K}^{T} \gamma^{i-t-K} r(i) ~|~   \pi_f(\boldsymbol{x}|\boldsymbol{h},\boldsymbol{y},\omega) \right), 
	\end{aligned}
\end{equation}
which is the expected cumulative reward that the UAV-IR can achieve by serving the UE from location $\boldsymbol{x}(t+K)$ during the next communication stage, given a current CSI   $\boldsymbol{h}(t)$ and a current  UE location $(\boldsymbol{y}(t),\omega(t))$.  
The parameter $K$ is the duration of the mobility stage, such that $\| \boldsymbol{x}(t) - \boldsymbol{x}(t+K)\| = v_r^{\text{max}} \cdot K \cdot \Delta T $.     
$\gamma \in [0,1]$ is the discount ratio of a future reward in the current estimation, and its value  depends on the prediction accuracy of $f$ to forecast the future movement of the served UE.  
Under a perfect prediction, $\gamma=1$ and the value function $Q$ will be equivalent to the objective function in (\ref{equsOpt}).   
Given the value function in (\ref{Qfunction}), the optimal policy $\pi$ to maximize the expected cumulative rewards, is given as 
\begin{equation}\label{optLocation}
	\boldsymbol{x}(t+K) = \arg \max_{\boldsymbol{x}}  Q_t(\boldsymbol{h}(t),\boldsymbol{x},\boldsymbol{y}(t),\omega(t)), 
\end{equation}  
which shows the optimal location where the UAV-IR should be deployed during the next communication stage, based on the current environment information. 

Once the UAV-IR arrives at $\boldsymbol{x}(t+K)$, the UE's location $(\boldsymbol{y}(t+K),\omega(t+K))$ and its movement pattern  $f(\cdot|\boldsymbol{Y}_{t+K})$ will be updated. 
Then, the UAV-IR can measure the CSI $\boldsymbol{h}(t+K)$, and set the reflection parameter $\Theta^{*}(t+K)$ according to Proposition \ref{prop1}. Consequently, a reward $r(t+K)$ will be given to the UAV-IR. 
The new information about the environment can be used to update the value function, based on the Q-learning algorithm, via  
\begin{equation}
	\begin{aligned}
		Q_{t+K} & ( \boldsymbol{h}(t), \boldsymbol{x}(t+K), \boldsymbol{y}(t), \omega(t) ) =   \\
		& (1-\beta) ~ Q_t (\boldsymbol{h}(t), \boldsymbol{x}(t+K), \boldsymbol{y}(t), \omega(t)) +   \beta ~   [ r(t+K) \\
		& + \gamma \max_{\boldsymbol{x}} Q_{t}(\boldsymbol{h}(t +K), \boldsymbol{x}, \boldsymbol{y}(t+K), \omega(t+K) ) ],  
	\end{aligned}
\end{equation}
where $\beta \in(0,1] $ is the learning rate.

However, given an infinite value space of $(\boldsymbol{h}, \boldsymbol{x}, \boldsymbol{y}, \omega)$, it is impossible to define the value function $Q$ for any possible tuple by measurement. 
Therefore, our goal is to find a  parametrized value function $\tilde{Q}(\cdot|\boldsymbol{\phi})$ to approximate the real value function 
$Q(\cdot)$.  
In order to find a good-enough approximation, a deep neural network, based on the long short-term memory architecture \cite{chen2019artificial}, can be applied to train the parameter $\boldsymbol{\phi}$.  
Note that, it is possible for the Q-learning algorithm to converge to a local optimal value function $Q(\cdot|\phi^*)$, if the service time $T$ of the UAV-IR is long enough and the UE's movement pattern is properly modeled \cite{sutton2018reinforcement}.  
The deep RL algorithm for the efficient deployment of the UAV-IR  is summarized in Algorithm \ref{algo}.

\begin{algorithm}[h] \small   
	\caption{Deep RL Algorithm for the UAV-IR deployment} \label{algo}
	\begin{algorithmic}
		\State \textbf{Initialize} the location $\boldsymbol{x}_0$, the onboard energy $E_0$,  the value \\
		\quad  function $\tilde{Q}(\boldsymbol{h}, \boldsymbol{x}, \boldsymbol{y}, \omega|\boldsymbol{\phi}_0)$, and the prediction   model  $f_0(\boldsymbol{y},\omega)$. \\
		\textbf{For} $t = 0, \cdots, T$\\
		\quad A. If $v_r(t)>0$, UAV continues moving; \\
		\quad B .If $v_r(t)=0$, \\ 
		\quad \quad 1. Update the UE's location $(\boldsymbol{y}(t), \omega(t))$,  the prediction   model \\
		\quad  \quad \quad  $f(\boldsymbol{y},\omega|\boldsymbol{Y}_{t})$,  and measure the CSI $\boldsymbol{h}(t)$. \\ 
		\quad \quad 2. Based on Proposition \ref{prop1}, optimize the reflection parameter  \\
		\quad  \quad \quad   
		by	$\Theta(t)=\arg \max_{\Theta} c(\boldsymbol{x}(t),\boldsymbol{y}(t),\omega(t),\Theta)$.   \\ 
		\quad \quad 3. Receive the reward $r(t)$. If $\eta(t)\ge \tau$, hover and $\boldsymbol{x}(t+1) = $  \\
		\quad  \quad \quad    
		$\boldsymbol{x}(t)$; Otherwise, UAV will move to \\ 
		\quad \quad \quad $ \boldsymbol{x}(t+K)= \arg \max_{\boldsymbol{x}}  \tilde{Q} (\boldsymbol{h}(t),\boldsymbol{x},\boldsymbol{y}(t),\omega(t) | \boldsymbol{\phi}(t) )$.   \\ 
		\quad \quad 4. If $E(t)<\epsilon$ or $\eta(t)\ge \tau$,   $\tilde{\boldsymbol{x}} = \boldsymbol{x}(t)$ and $q(t) = r(t)$; \\  
		\quad \quad \quad Otherwise, set $\tilde{\boldsymbol{x}} = \boldsymbol{x}(t+K)$, and\\
		\quad \quad \quad \quad $q(t) = r(t) + \gamma \max_{\boldsymbol{x}} \tilde{Q}(\boldsymbol{h}(t),\boldsymbol{x},\boldsymbol{y}(t),\omega(t)|\boldsymbol{\phi}(t))$. \\ 
		\quad \quad 5. Train the parameter $\boldsymbol{\phi}$ of the value function $\tilde{Q}$  using  a    \\
		\quad \quad \quad neural network,  to minimize the loss function: \\ 
		\quad \quad \quad $ \boldsymbol{\phi}(t+1) = \arg \min_{\boldsymbol{\phi}} [q(t) - \tilde{Q}(\boldsymbol{h}(t),\tilde{\boldsymbol{x}},\boldsymbol{y}(t),\omega(t)|\boldsymbol{\phi})]^2$. \\ 
		\quad \quad 6. Update the value function to be $\tilde{Q}(\cdot|\boldsymbol{\phi}(t+1))$. \\
		\quad C. Update onboard energy $E(t+1) = E(t) -p(\Theta(t),v_r(t))\Delta T$. \\	
		\textbf{Until} $E(t) < \epsilon$. 
	\end{algorithmic}
\end{algorithm}

\section{Simulation Results and Analysis}\label{simulation}

\begin{figure}[!t]
	\begin{center}
		\vspace{-0.5cm}
		\includegraphics[width= 9.3cm]{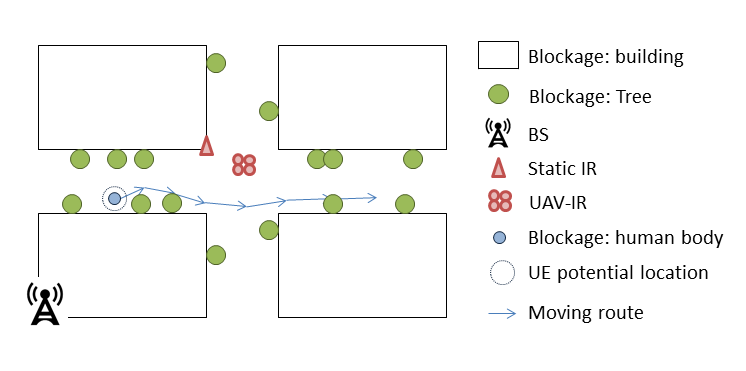}
		\vspace{-0.7cm}
		\caption{\label{simEnviron}\small Illustration of the considered deployment environment for our simulations.}
	\end{center}\vspace{-0.5cm}
\end{figure}

For our simulations, we consider a 3D urban environment that consists of two streets, four buildings, and trees that can obstruct mmW signals. 
As shown in Fig. \ref{simEnviron}, each building is of dimensions $40$ m $\times$ $40$ m $\times$ $16$ m, 
the street width is $15$ m, and each tree, distributed randomly on the street side, has a crown radius of $2$ m.       
The UE swings back and forth with respect to a human user, whose moving route is modeled by an MDP with a random destination.  
The BS-UE channel is blocked persistently by buildings. 
The path loss of mmW transmissions is based on the UMi-Street Canyon model \cite{5GCM}, and the mmW multi-input-multi-output channel is based on  \cite{ratnam2019continuous}.    
The carrier frequency is set to be $30$ GHz, the BS transmit power is $40$ dBm, with $b$ $=$ $ 0.1$ GHz, $N=16$, $M=64$, $V_r^{\text{max}} = 20$ m/s, $V_e = 1$ m/s, $a = 0.8$, $\kappa$ $=$ $60\%$, $\Delta T = 0.1$ s, $\tau = 5$ dB, and $\gamma =0.1$. 
The planar antenna array with $4 \times 4$ IR components and $8 \times 8$ antennas are equipped on the BS and the UAV-IR, respectively.  
For simplicity,  in our simulations, we assume a constant altitude for the UAV-IR. 
In order to examine the efficiency of the UAV-IR with an RL-based deployment, two other IRs are  introduced in our simulations: a static IR placed on the top of a building near the crossing with a height of $16$ m, and a UAV-IR without RL, which moves towards the UE's latest location, every time downlink blockage occurs.   

\begin{figure}[!t]
	\begin{center} 
		\includegraphics[width= 9.3cm]{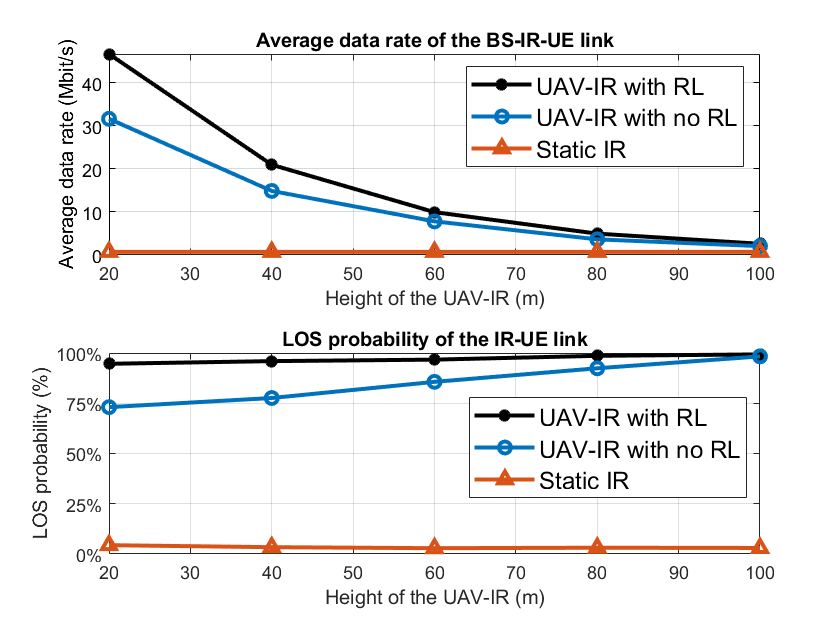}
		\vspace{-0.5cm}
		\caption{\label{height}\small  Average data rate per time slot of the BS-IR-UE link and the average LOS probability of the  IR-UE link.}
	\end{center}\vspace{-0.6cm}
\end{figure}

Fig. \ref{height} shows the average data rate of the BS-IR-UE link, as the UAV's height increases from $20$ to $100$ m. 
At an altitude of $20$ m, for both scenarios (with or without RL), the UAV-IRs yield a significant improvement in average data rate, compared to the static case.   
However, as the altitude increases, the performance of  UAV-IR decreases due to the increasing path loss at the BS-IR-UE channels.  
At a higher elevation, although it is more likely to have a LOS link,  the mmW channel attenuation becomes more severe as the BS-UAV and UAV-UE distances increase.   
As a result, when the height exceeds $100$ m, the performance of  UAV-IR is similar to the static case. 
Compared with the RL-based deployment, the lower data rate of the UAV-IR with no RL is caused by the frequent movement, due to its short-sighted deployment strategy.   
From Fig. \ref{height}, we can also see that a static IR   results in a LOS probability lower than $5\%$.  
Owing to mobility, the UAV-IR with the RL-based deployment yields a LOS probability over $90\%$, while the UAV-IR with no RL shows a probability above $70\%$.  
The performance improvement of the learning-based UAV-IR over the UAV with no RL is due to the prediction of the UE's movement and the estimation of each location over a long term.  

\begin{figure}[!t] 
	\begin{center}
		\vspace{-0.2cm}
		\includegraphics[width= 9.3cm]{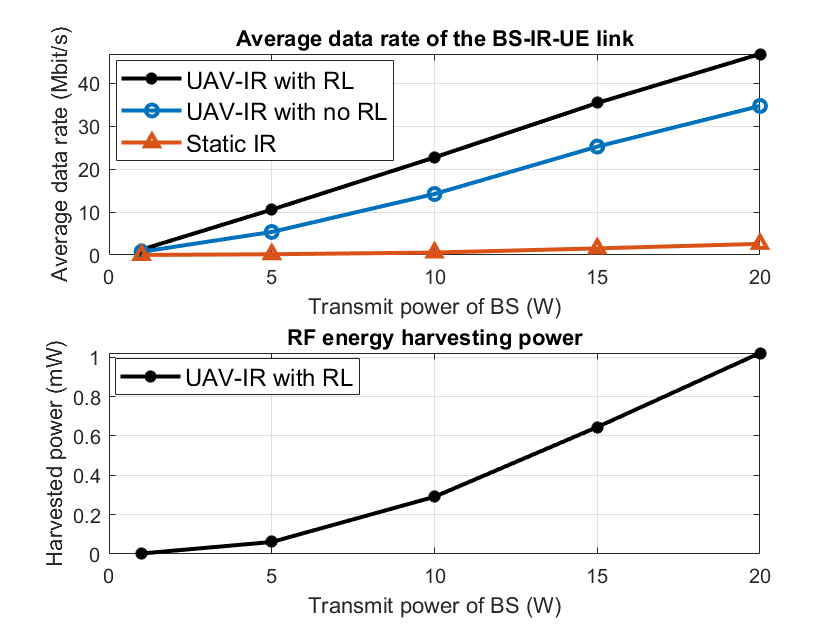}
		\vspace{-0.4cm}
		\caption{\label{power}\small Average data rate of the BS-IR-UE link and the RF energy harvesting power, with a fixed UAV height at $40$ m.  }
	\end{center}\vspace{-0.6cm}
\end{figure}

Fig. \ref{power} shows that the average data rate of the BS-IR-UE link,  at a UAV height of $40$ m, as the BS transmit power increases from $1$ to $20$ W.   
For a lower transmit power, the downlink SNR of the UAV-reflected channel will be smaller than the threshold $\tau$, which results in a zero data rate.  
In this case, although the downlink channel states of UAVs are better than the static IR, the performance of the three cases is similar. 
As the BS transmit power increases, the data rates increases in all scenarios.  
However, the downlink rates of UAV-IRs increase much faster than the static IR. Moreover, due to a higher LOS probability, the data rate of the RL approach increases faster than the UAV-IR with no RL.  
From Fig. \ref{power}, we can also see that, as the BS transmit power increases, the RF energy harvested by the UAV-IR increases up to $1$ mW, which is adequate for  IR to self-power without drawing any energy from the UAV.  

\section{Conclusion}\label{conclusion}
In this paper, we have proposed a novel RL-based approach to enable an efficient deployment of the UAV-IR in serving the downlink mmW transmission to a mobile outdoor UE, with energy harvesting in the IR. 
To maintain a downlink LOS channel, the propagation environment is modeled, via Q-learning and neural networks, to optimize the location and reflection coefficient of the UAV-IR, such that the downlink communication capacity can be  maximized.  
Simulation results show  a significant advantage of using a UAV-IR compared to a static IR, in terms of the average data rate and the achievable downlink LOS probability.   
The results also show that the RL-based deployment of the UAV-IR further improves the network performance, relative to a UAV-IR without learning.    
Moreover, the RF energy harvested from mmW beamforming signals is shown to be sufficient to power the IR for signal reflections.

\bibliographystyle{IEEEtran}
\bibliography{references}

\end{document}